\documentclass{article}

\usepackage{PRIMEarxiv}

\usepackage[T1]{fontenc}    
\usepackage{hyperref}       
\usepackage{url}            
\usepackage{booktabs}       
\usepackage{amsfonts}       
\usepackage{nicefrac}       
\usepackage{microtype}      
\usepackage{fancyhdr}       
\usepackage{graphicx}       
\usepackage[percent]{overpic}
\usepackage{amsthm}
\usepackage{amsmath}
\usepackage{xcolor}

\newcommand{\name}{NerualMLS}

\newtheorem{prop}{Proposition}

\title{NerualMLS: Geometry-Aware Control Point Deformation
}

\author{
  Meitar Shechter$^{1}$, Rana Hanocka$^{2}$, Gal Metzer$^{1}$, Raja Giryes$^{1}$, Daniel Cohen-Or$^{1}$\\ \\
  $^1$Tel-Aviv University, Israel, 
  $^2$University of Chicago, USA \\
}

\begin{document}
\twocolumn[

\maketitle
\hrule
\bigskip
\begin{abstract}
We introduce \name, a space-based deformation technique, guided by a set of displaced control points. 
We leverage the power of neural networks to inject the underlying shape geometry into the deformation parameters.
The goal of our technique is to enable a realistic and intuitive shape deformation.
Our method is built upon moving least-squares (MLS), since it minimizes a weighted sum of the given control point displacements. 
Traditionally, the influence of each control point on every point in space (\emph{i.e.,} the weighting function) is defined using inverse distance heuristics. In this work, we opt to \textit{learn} the weighting function, by training a neural network on the control points from a single input shape, and exploit the innate smoothness of neural networks.
Our geometry-aware control point deformation is agnostic to the surface representation and quality; it can be applied to point clouds or meshes, including non-manifold and disconnected surface soups.
We show that our technique facilitates intuitive piecewise smooth deformations, which are well suited for manufactured objects. 
We show the advantages of our approach compared to existing surface and space-based deformation techniques, both quantitatively and qualitatively.
\end{abstract}
\bigskip
\hrule
\bigskip


\bigskip
]
\begin{figure}
    \centering
    \includegraphics[width= 0.7\columnwidth]{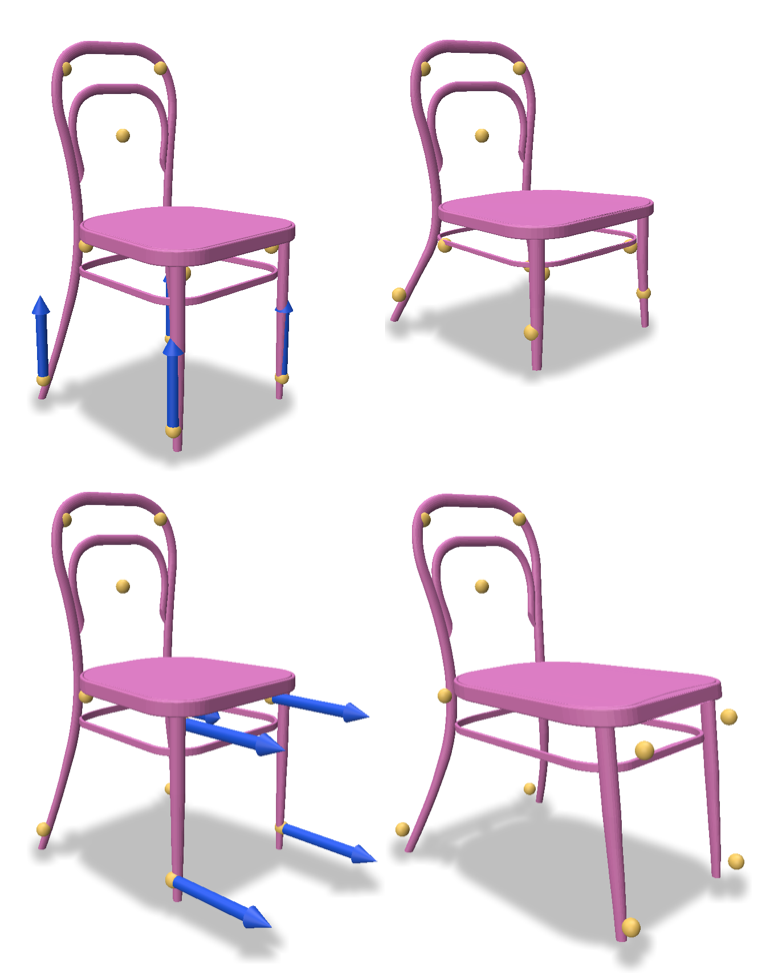} 
    \caption{\emph{Interactive shape deformation using \name.} We train a neural network on a given set of control points (yellow), in a displacement-agnostic manner. Then, the user applies various displacements and receives different deformations.}
\label{fig:teaser}
\end{figure}

\section{Introduction}
Interactive deformation of 3D shapes is a long-standing problem in computer graphics \cite{10.2312:egt.20091068, yuan2021revisit}. 
The objective is to provide the user a simple and intuitive interface for semantic shape editing and manipulation.
The surge of successful use of neural networks on images~\cite{dosovitskiy2021image, mao2021voxel, dai2021coatnet} and text~\cite{vaswani2017attention, brown2020language, radford2021learning}, has garnered interest in applying them to the field of geometric modeling. Yet, exploiting neural networks to inject semantic understanding into shape deformation is only in its infancy. 
The work presented in this paper is a step forward along that direction.

We present a space-deformation technique, where the user manipulates a set of control points which are interactively displaced to define the deformation. The technique is built upon moving least-squares (MLS), as it minimizes a weighted sum of the given control point displacements. 
However, in traditional MLS the weighting function is crafted based on pre-defined heuristics, which are agnostic to the given input shape. Instead, we opt to \textit{learn} a weighting function, which is \textit{tailored} to the given input shape using a neural network, as illustrated in Figure~\ref{fig:2d_illustration}, leading to a \name.

Given a set of control points, we train a network to learn the corresponding (control point) weighting function. We assign each control point a unique \emph{classification label}, and train a network to learn to map each control point position to the assigned 1-hot encoding. During inference, we obtain the \textit{influence} of each control point for an arbitrary point in space through a feed-forward pass of the network. This yields a set of network-predicted probabilities for the input point to belong to each control point. 

Our technique benefits from the innate inductive biases of neural networks (e.g., they tend to learn the low frequencies first \cite{pmlr-v97-rahaman19a, basri2020frequency, Bietti2019OnTI}, and intrinsically favor smooth and desirable solutions~\cite{ulyanov2018deep, 10.1145/3386569.3392415, metzer2020self}). 
The inductive bias of neural networks enables us to train the network to learn a smooth weighting function from only the given control points, and avoids having to train on a large dataset.
Furthermore, our technique benefits from the tendency of neural networks to converge to a piecewise smooth solution~\cite{Savarese19How,Williams19Gradient,Ongie2020Function,giryes2020function}.
Piecewise smooth deformations are desirable when editing and manipulating manufactured objects. It is particularly useful for keypoint-based deformations, since we want to deform or displace each part separately while smoothly interpolating the geometric details within each part independently.

The interface of our method and a set of shape manipulations are displayed in Figure~\ref{fig:overview}. Our method:
(i) provides an intuitive and simple framework for shape manipulation; 
(ii) requires no training dataset; 
(iii) produces piecewise smooth deformations; and 
(iv) provides control over the relaxation degree of the interpolation property.
We demonstrate the advantages of our approach compared to existing surface and space-based deformation techniques, both quantitatively and qualitatively. Our code is available at \url{https://github.com/MeitarShechter/NeuralMLS}.

\begin{figure}
    \centering
    \newcommand{\pl}{34}
    \begin{overpic}[width=\columnwidth]{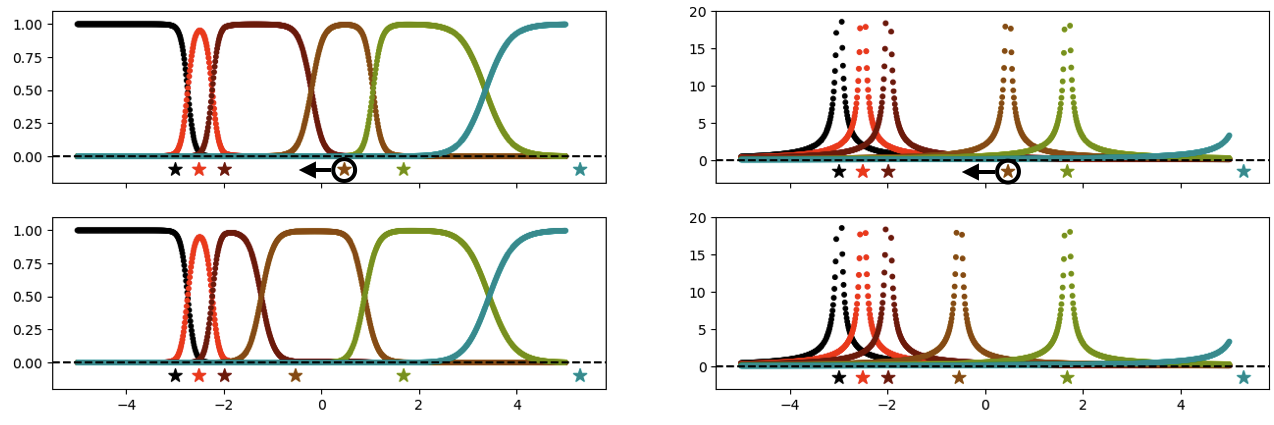}
    \put(15,  \pl){\textcolor{black}{\name}}
    \put(69,  \pl){\textcolor{black}{MLS\cite{Zhu20073DDU}}}
    \put(100, 17){\rotatebox{270}{\small Config 2}}
    \put(100, 33){\rotatebox{270}{\small Config 1}} 
    \put(-3, 4){\rotatebox{90}{\textcolor{gray}{\footnotesize weight}}}
    \put(-3, 21){\rotatebox{90}{\footnotesize \textcolor{gray}{weight}}} 
    \put(15,  -3){\textcolor{gray}{\footnotesize point location}}
    \put(69,  -3){\textcolor{gray}{\footnotesize point location}}
    \end{overpic}
    \caption{\emph{\name learns a geometry-aware weighting function.} We show two different (1D) control point configurations 
    (denote by the star locations). We visualize the learned/calculated weighting function of each control point for both \name (left) and MLS (right) by color coding.
    Note the difference in the weight distribution of the 
    points which are adjacent to the moved 
    point in \name, compared to the lack of change in MLS.}
    \label{fig:2d_illustration}
\end{figure}
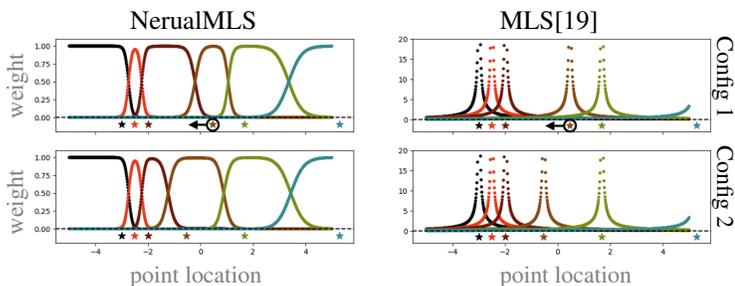



\section{Related Work}

Recently, there has been a rising interest in using deep learning to edit shapes by changing their geometry or their appearance through deformations \cite{liu2021editing, Mehr_2019_ICCV, uy-joint-cvpr21, wei2020learning}.
Our paper focuses on shape deformation which has been studied extensively in the literature for the past several decades \cite{10.2312:egt.20091068, yuan2021revisit}. One of the common schemes are referred to as space-based deformation, where the deformation is applied over the entire ambient space containing the shape. Among space-based methods are free-form deformation (FFD)~\cite{10.1145/15886.15903}), cage-based deformation (CBD)~\cite{10.1145/1186822.1073229, 10.1145/1276377.1276466, 10.1145/1360612.1360677} and moving least squares (MLS)~\cite{10.1145/1141911.1141920, Zhu20073DDU}.
As space-based deformations seek to interpolate a given set of point displacements over the entire space, the main challenge that arises is: how to define an interpolation that follows desired deformation properties, such as piecewise smoothness and detail preservation of the input shape.

Deep learning has been used for shape alignment (i.e., source to target deformations), to mitigate known limitations in existing hand-crafted techniques. For instance, Hanocka et al. \cite{hanocka2018alignet} learn a FFD grid which can warp source shapes to incomplete target shapes. Yifan et al. \cite{Yifan:NeuralCage:2020} use a neural network to learn a cage based deformation (CBD) and facilitate the cage construction task.
Other works harness neural networks to learn a per-point displacement \cite{wang20193dn, groueix19cycleconsistentdeformation}.
Another limitation with the above methods is their inability to provide a simple and intuitive interface for shape manipulation. Constructing a cage is a rather complex task, where as annotating a few control points is easy and straight-forward.

Shape deformation based on a sparse set of control points provides an intuitive interface for shape manipulation. Control point based methods enable the user to guide the deformation using various constraints. In As-Rigid-As-Possible (ARAP)~\cite{10.2312:SGP:SGP07:109-116}, the control points impose hard constraints on a distortion minimization objective. In KeypointDeformer (KPD)~\cite{jakab2021keypointdeformer}, the control points are used to influence a cage enclosing the shape. Moving least squares (MLS)~\cite{10.1145/1141911.1141920, Zhu20073DDU} uses the control points to construct a least squares minimization problem.
Although control point based methods often yield desirable results, they still suffer from various drawbacks.
ARAP~\cite{10.2312:SGP:SGP07:109-116} is impractical for large meshes and requires high quality, manifold, single-connected component mesh, which is a considerable demand especially for scanned objects. KPD~\cite{jakab2021keypointdeformer} is required to be trained on the same class used at inference time, and only provides a fixed amount and initial placement of control points for the user.
MLS~\cite{Zhu20073DDU} might suffer from local artifacts around the control points and over smoothing of rigid parts, as seen in Figures~\ref{fig:ShapeDeformationResults} and~\ref{fig:MLS_ablation_alpha}.


The spectral bias and smoothness properties of neural networks studied in \cite{pmlr-v97-rahaman19a, basri2020frequency, Bietti2019OnTI} suggest that networks learn low frequencies first and provide a type of smooth interpolation over the training data. 
This can be a desired property for many applications in computer graphics and vision, as demonstrated by the self-prior concept when training on a single or few self-supervised examples,
such as in various image tasks~\cite{8579082}, mesh reconstruction~\cite{10.1145/3386569.3392415, Wei2021DeepHS} and point cloud consolidation~\cite{metzer2020self}.
Furthermore, the works of~\cite{Savarese19How,Williams19Gradient,Ongie2020Function,giryes2020function} showed that the types of functions learned by neural networks tend to be piecewise smooth, which is particularly interesting for shape manipulation, as it enables arbitrarily control on the part level while smoothly interpolating details within each part separately. 
This type of deformations is especially difficult to produce with previous closed-form hand-crafted techniques, that typically struggle with the inherent tradeoff between smooth and piecewise, as can be seen in Figures~\ref{fig:ShapeDeformationComparison} and~\ref{fig:ShapeDeformationResults}.



\begin{figure*}
    \centering
    \includegraphics[width=\textwidth]{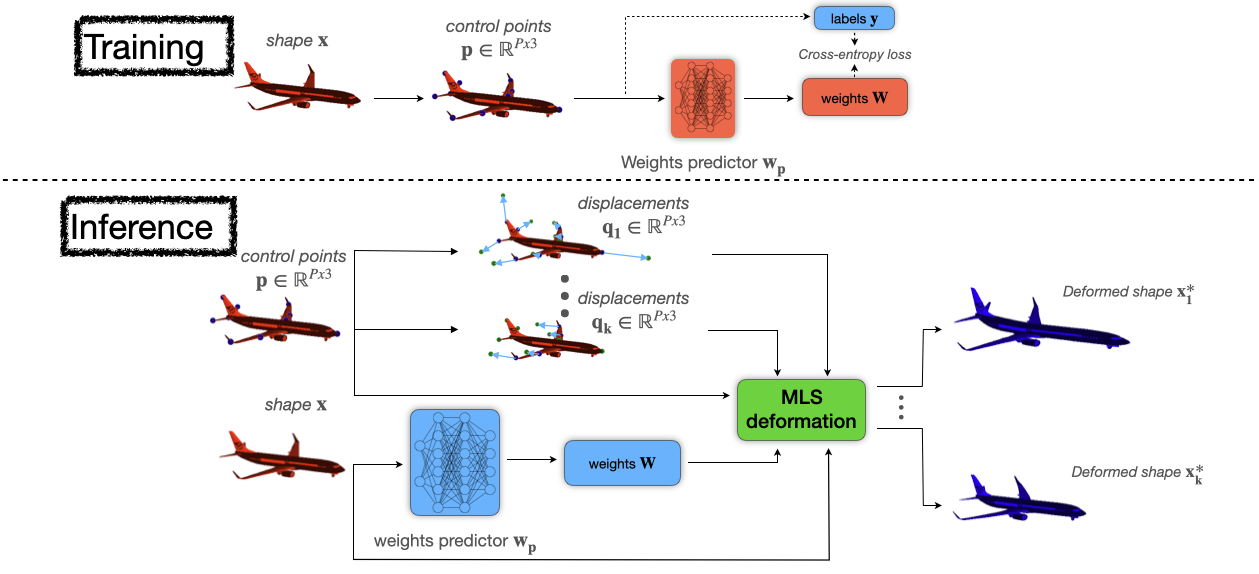}
    \caption{\emph{Method overview.} \textbf{Training (top)}: starting with a shape $\mathbf{x}$ and given control points, we train a weighting function network $\mathbf{w_p}$ to learn to map each control point position to an assigned 1-hot encoding.  
    \textbf{Inference (bottom)}: 
    during inference, we freeze the network weights and each point in the shape $\mathbf{x}$ passes through $\mathbf{w_p}$ which generates the per-point weights $\mathbf{W}$. The user can then interactively displace the control points to new locations $\mathbf{q}$, which will generates a different shape $\mathbf{x^*}$ using an MLS deformation.}
    \label{fig:overview}
\end{figure*}
%

\section{Method}

Our method, \name, learns a weighting function, which is then used within the MLS framework. Our technique is self-supervised in the sense that it is trained on a single input shape and a set of user annotated control-points.
The learned weighting function turns MLS from an interpolation method into an approximation method, in the sense that each control is a soft constraint rather than a hard one, yielding a desirable approximation that is piece-wise smooth.

To learn the weighting function, we treat each control point as a class, and train a neural network to output the corresponding class, based on the control point's xyz position in space.
Then, given a new arbitrary point in space as input, the network assigns soft-max probabilities to each control-point (class) that are used as weights for MLS. 


In the following, we elaborate more on the technique. We first give a brief overview on plain MLS, and then continue to explain how we train and use the network.

\subsection{Moving Least Squares (MLS)}
Moving-Least-Squares is a type of space-deformation techniques, where rather than exclusively deforming the source shape $S$, the entire ambient space is deformed. In MLS, the set of the given control points and their displacements, denoted by $\mathcal{P}$ and $\mathcal{Q}$, are projected into the entire space.
Specifically, it finds for any given point in space $\mathbf{x}\in \mathbb{R}^3$, the deformation  $T_x(\cdot)$ that minimizes the following objective:
\begin{equation}
\label{eq:1}
    E = \sum_i w_i(\mathbf{x})|T_x(\mathbf{p_i}) - \mathbf{q_i} |^2,
\end{equation}
where $\mathbf{p_i}$ are the control points original locations, $\mathbf{q  _i}$ are the control points new locations and $w_i$ are the weighting schemes of each control point. 

The common weighting function used is an Euclidean-based one:
\begin{equation} \label{eq:2}
    w_i(\mathbf{x}) = \frac{1}{d(\mathbf{p_i}, \mathbf{x})^{2\alpha}},
\end{equation}
where $d$ is the distance metric (Euclidean in MLS) and $\alpha$ is a fall-off parameter that weigh the distance function.
A close-form solution for $T_x$ can be obtained when the space of possible deformations is constrained to be an affine, similarity or rigid deformation \cite{10.1145/1141911.1141920, Zhu20073DDU}.

\subsection{Learning a Geometry Aware Weighting Function}
The MLS framework can be seen as a deformation function $f: \mathbb{R}^3 \times \mathbb{R}^{Px3} \times \mathbb{R}^{Px3}
\times w_p\rightarrow \mathbb{R}^3$, where the inputs are a point in space $\mathbf{x}$, 
the control points $\mathcal{P}$, the control points after displacement $\mathcal{Q}$ and a weighing function  $w_p:\mathbb{R}^3\rightarrow \mathbb{R}_{\geq 0}^{P}$, and the output is the deformed point location. For performing the deformation, we apply $f$ on all the points of the given shape.
Our work focuses on learning a new weighting function, which we plug into the MLS framework and gain the following desired deformation properties, for the deformation function $f$:
\begin{enumerate}
    \item \emph{Approximation}: the control points $\mathcal{P}$ should map \emph{approximately} to their displaced locations $\mathcal{Q}$, i.e., $f(\mathbf{p_i})\approx \mathbf{q_i}$.
    \item \emph{piecewise Smoothness}: $f$ should produce piecewise smooth deformations.
    \item \emph{Identity}: If the control points have not been displaced then $T_x$ should be the identity function, i.e., $f(\mathbf{x}) = \mathbf{x}$.
\end{enumerate}
As the third property is guaranteed by the principles of MLS, our effort concentrate on the first two.

Our weighting function ${w_p}: \mathbb{R}^3 \rightarrow \mathbb{R}_{\geq 0}^{P}$ is computed using a simple multi-layer perceptron (MLP), which consists of 2 hidden layers with a ReLU activation, along with a softmax normalization layer in order to get non-negative weights. 

We use a proxy classification task, where each given control point $\mathbf{p_i}\in\mathcal{P}$ is considered as a class. We train our network by minimizing a categorical cross-entropy classification loss such that the input is the control point and its label is its corresponding class. We find this type of training with a neural network a proper fit due to the piecewise smooth interpolation property of neural networks \cite{Savarese19How,Williams19Gradient,Ongie2020Function,giryes2020function} and the ability to control the approximation degree via a simple temperature parameter within the softmax normalization. As a result, the network provides us with a weight for each point, where the weight function has a piecewise smooth structure due to the network inherent bias.

Once the control points are given, our method has a warm-up phase, i.e., the training of the network, where we first construct the appropriate network by adjusting its output layer to be of the size of the number of control points and then train it on this set of control points (notice that the training is independent of the control points displacements). As the number of network parameters and training examples is small, this phase takes several seconds even on a personal laptop. Now, as we have our weighting function in hand, the user can manipulate the control points in order to interactively deform the shape, as illustrated in Figure~\ref{fig:overview}.

By posing our problem as a space classification one, we implicitly inject awareness to the underlying control point geometric configuration, as each control point also acts as a negative example for all the others. This makes the learned weight function depend on an entire control point configuration, rather than being affected by each control point independently, as illustrated in Figure~\ref{fig:2d_illustration}.
This property, for rigid deformations, is not manifested in the Euclidean-based weighting function (or any hand-crafted weighting scheme that is independent of the rest of the control points), as can be concluded from the following proposition that shows that MLS output does not change if we normalize the weights it is using:

\begin{prop}
Let $\mathbf{x}$ be a point in space, $\mathcal{P}$ a set on control points, $\mathcal{Q}$ be their displacements and $w$ the Euclidean-based weighting function, i.e. $w_i(\mathbf{x}) = \frac{1}{d(\mathbf{p_i},\mathbf{x})^{2\alpha}}$. Let $T_x=MLS(\mathbf{x},\mathcal{P},\mathcal{Q},w)$ be the deformation function obtained by plugging $\mathbf{x}$, $\mathcal{P}$, $\mathcal{Q}$ and $w$ to the MLS framework with rigid-deformations constraint, namely, $T_x(\mathbf{y})=M\mathbf{y}+\mathbf{r}$ is the optimal rigid solution for some rotation matrix $M$ and translation vector $\mathbf{r}$. Also, let $\hat{w}$ be the normalized weights, i.e., $\hat{w}_i(\mathbf{x}) = \frac{w_i(\mathbf{x})}{\sum_i w_i(\mathbf{x})}$, and $\hat{T_x}=MLS(\mathbf{x},\mathcal{P},\mathcal{Q},\hat{w})$ be the optimal rigid solution of the MLS framework with the normalized weights.
Then $\hat{T_x}(\mathbf{y})=\hat{M}\mathbf{y}+\mathbf{\hat{r}}\equiv T_x(\mathbf{y})$.
\label{claim:1}
\end{prop}

\begin{proof}
Our proof follows similar transitions and notations as in \cite{Zhu20073DDU} and \cite{10.1145/1141911.1141920}, but for completeness we attach the full proof here.\\
Let $\mathbf{p_*}$ and $\mathbf{q_*}$ be the weighted centroids of $\mathbf{p_i}$'s and $\mathbf{q_i}$'s, respectively:
\begin{equation*}
    \mathbf{p_*} = \frac{\sum_i w_i(\mathbf{x})\mathbf{p_i}}{\sum_i w_i(\mathbf{x})}, \hspace{3mm} \mathbf{q_*} = \frac{\sum_i w_i(\mathbf{x})\mathbf{q_i}}{\sum_i w_i(\mathbf{x})}
\end{equation*}
Similarly, define $\mathbf{\hat{p}_*}$ and $\mathbf{\hat{q}_*}$ as the weighted centroids using the normalized weights. Then we have:
\begin{equation*}
    \mathbf{\hat{p}_*} = \frac{\sum_i \hat{w}_i(\mathbf{x})\mathbf{p_i}}{\sum_i \hat{w}_i(\mathbf{x})} = \sum_i \hat{w}_i(\mathbf{x})\mathbf{p_i} = \sum_i\frac{w_i(\mathbf{x})\mathbf{p_i}}{\sum_j w_j(\mathbf{x})} = \mathbf{p_*}
\end{equation*}
And similarly, $\mathbf{\hat{q}_*} = \mathbf{q_*}$.\\
By plugging $T_x(\mathbf{y})=\hat{T_x}(\mathbf{y})=\hat{M}\mathbf{y}+\mathbf{\hat{r}}$ into equation~\ref{eq:1} we get a quadratic dependency in $\mathbf{\hat{r}}$. Since the minimizer is where the derivatives with respect to each of the free variables in $\hat{T_x}$ are zero, we can solve directly for $\mathbf{\hat{r}}$ in the terms of the matrix $\hat{M}$. Taking the partial derivative w.r.t. the free variables in $\mathbf{\hat{r}}$ produces a linear system of equations. Solving for $\mathbf{\hat{r}}$ yields that 
\begin{equation*}
    \mathbf{\hat{r}}=\mathbf{\hat{q}_*}-\mathbf{\hat{p}_*}\hat{M}=\mathbf{q_*}-\mathbf{p_*}\hat{M}
\end{equation*}
This leaves only $\hat{M}$ to be determined. Note:
\begin{equation*}
\begin{split}
    \alpha_i=w_i(\mathbf{x})^{\frac{1}{2}}, \hspace{3mm} \mathbf{\tilde{p}_i} = \mathbf{p_i} - \mathbf{p_*}, \hspace{3mm} \mathbf{\tilde{q}_i} = \mathbf{q_i} - \mathbf{q_*}, \hspace{1mm}\\
    P = (\alpha_1\mathbf{\tilde{p}_1} \cdots \alpha_N\mathbf{\tilde{p}_N}), \hspace{3mm} Q = (\alpha_1\mathbf{\tilde{q}_1} \cdots \alpha_N\mathbf{\tilde{q}_N})
\end{split}    
\end{equation*}
Similarly, we note for the normalized weights:
\begin{equation*}
\begin{split}
    \hat{\alpha}_i=\hat{w}_i(\mathbf{x})^{\frac{1}{2}}, \hspace{3mm}
    \mathbf{\tilde{\hat{p}}_i} = \mathbf{p_i} - \mathbf{\hat{p}_*} = \mathbf{\tilde p_i}, \hspace{3mm}
    \mathbf{\tilde{\hat{q}}_i} = \mathbf{q_i} - \mathbf{\hat{q}_*} = \mathbf{\tilde q_i} \\
    \hat{P} = (\hat{\alpha}_1\mathbf{\hat{\tilde{p}}_1} \cdots \hat{\alpha}_N\mathbf{\hat{\tilde{p}}_N}), \hspace{3mm} \hat{Q} = (\hat{\alpha}_1\mathbf{\hat{\tilde{q}}_1} \cdots \hat{\alpha}_N\mathbf{\hat{\tilde{q}}_N}) \hspace{5mm}
\end{split}    
\end{equation*}
Then
\begin{equation*}
\begin{split}
    E = \sum_i \hat{w}_i(\mathbf{x})|\hat{M}\mathbf{\hat{\tilde{p}}_i}-\mathbf{\hat{\tilde{q_i}}}|^2 =
    \sum_i |\hat{M}\hat{\alpha_i}\mathbf{\hat{\tilde{p}}_i}-\hat{\alpha_i}\mathbf{\hat{\tilde{q_i}}}|^2 = \\
    ||\hat{M}\hat{P}-\hat{Q}||^2_F = tr((\hat{M}\hat{P}-\hat{Q})^t(\hat{M}\hat{P}-\hat{Q})) = \\
    tr(\hat{P}^t\hat{P})+tr(\hat{Q}^t\hat{Q})-2tr(\hat{Q}^t\hat{M}\hat{P}) \hspace{16mm}
\end{split}    
\end{equation*}
where $||\cdot||_F$ is the Frobenius norm. Since $\hat{P}$ and $\hat{Q}$ are constant, minimizing $E$ corresponds to maximizing $\psi= tr(\hat{Q}^t\hat{M}\hat{P}) = tr(\hat{M}\hat{P}\hat{Q}^t)$.\\
Now, observe that the following holds:
\begin{equation*}
\begin{split}
    \hat{P}\hat{Q}^t=
    \sum_{i}\hat{\alpha}_i \hat{\alpha}_i \mathbf{\tilde{\hat{p}}_i} (\mathbf{\tilde{\hat{q}}_i})^t =
    \sum_{i}\hat{\alpha}_i \hat{\alpha}_i \mathbf{\tilde{p}_i} (\mathbf{\tilde{q}_i})^t = 
    \sum_{i}\hat{w}_i(\mathbf{x}) \mathbf{\tilde{p}_i} (\mathbf{\tilde{q}_i})^t = \\
    \sum_{i}(\frac{w_i}{\sum_k w_k} \mathbf{\tilde p_i} (\mathbf{\tilde q_i})^t) = 
    \frac{1}{\sum_k w_k}\sum_{i}w_i \mathbf{\tilde p_i} (\mathbf{\tilde q_i})^t = \hspace{14mm} \\
    \frac{1}{\sum_k w_k}\sum_{i}\alpha_i \alpha_i \mathbf{\tilde p_i} (\mathbf{\tilde q_i})^t = 
    [c\equiv \frac{1}{\sum_k w_k}] =
    c \sum_{i}\alpha_i \alpha_i \mathbf{\tilde p_i} (\mathbf{\tilde q_i})^t = cPQ^t
\end{split}
\end{equation*}
Therefore, the singular value decomposition of $\hat{P}\hat{Q}^t = \hat{U}\hat{\Lambda}\hat{V}^t$ satisfies $\hat{U} = U, \hat{V}=V$ and $\hat{\Lambda} = c\Lambda$, where $U\Lambda V^t$ is the singular value decomposition of $PQ^t$. Thus
\begin{equation*}
    \psi = tr(\hat{M}cPQ^t) = tr(\hat{M}cU\Lambda V^t) = tr(U^t\hat{M}^tVc\Lambda)
\end{equation*}
Write $N=U^t\hat{M}^t V$, then $N$ is orthogonal since $U$, $\hat{M}$ and $V$ are orthogonal. It follows that $|N_{i,j}| \leq 1$, and
\begin{equation*}
    \psi = tr(Nc\Lambda) = \sum_{i=1}^{3} N_{i,i}c\lambda_i \leq \sum_{i=1}^{3} c\lambda_i
\end{equation*}
Hence, $\psi$ is maximized when $N=I \iff \hat{M}=VU^t=M$
and therefore $\mathbf{\hat{r}} = \mathbf{q_*}-\mathbf{p_*}M = \mathbf{r}$.
\end{proof}

From Proposition~\ref{claim:1}, we conclude that for rigid solutions, observing the weights $w_i$ with no normalization is enough as the solutions are the same, and therefore adding a new control point or manipulating others has so effect on the \emph{weight} of a certain control point. Thus, the weight of each control point is \emph{not} adaptive to the specific control point configuration.

\subsection{Approximating vs. Interpolating}
As we use a softmax normalization layer in our construction of the weighting function, we can add a temperature scaling parameter to the network's output, i.e.
\begin{equation}
    [w_p(\mathbf{x})]_i = \frac{e^{\frac{w_i(\mathbf{x})}{T}}}{\sum_i e^{\frac{w_i(\mathbf{x})}{T}}},
\end{equation}
where $w_i(\mathbf{x})$ are the network's outputs.

The temperature enables the user to control the degree of approximation versus interpolation.
As setting $T\rightarrow 0$ result in a \textit{sharper} weight distribution i.e. making the weight of the most dominant control point  approach $1$ and the other weights approach $0$.
In Section~\ref{ss:ApproVsInter}, we show the impact of this temperature parameter. 


\section{Experiments}
\label{sec:exp}

In this section we demonstrate the key benefits of our method, compared to plain MLS and other existing techniques. As part of the experiments, we show that our method is able to achieve more intuitive and realistic shape deformations  both qualitatively and quantitatively (Section~\ref{ss:shapeDeformation}), explore our more piece-wise smooth weighting function (Section~\ref{ss:piecewieseDeformation}), and demonstrate control over the approximation degree (Section~\ref{ss:ApproVsInter}).
We also include dedicated experiments showcasing the limitations of the Euclidean-based weighting function we propose to replace (Section~\ref{ss:ablation}).

\subsection{Experimental Setup}

\textbf{Data.} In the experiments, we demonstrate our performance abilities on data from different categories. Our method require no training data besides the user specified control points and input shape to deform.
In order to compare to KeypointDeformer~\cite{jakab2021keypointdeformer} we trained their method on ShapeNet~\cite{shapenet2015}, as it can not be applied on a single input.

\textbf{Implementation details.} Our weighting function $\mathbf{w_p}$ is implemented as a neural network using a simple multi-layer perceptron (MLP) consisted of 2 hidden layers of width 1024 and ReLU activation. The output layer size is set to the amount of user-given control points. We use the Adam optimizer~\cite{kingma2017adam}, and
unless stated otherwise, all results are attained using a soft-max temperature of $1$ and constraining the MLS solution to rigid deformations.

\textbf{Control points annotation.} 
All methods we compare to, except KeypointDeformer~\cite{jakab2021keypointdeformer}, require and allow control points that are manually annotated by the user.
In the evaluations against KeypointDeformer, we use the control points given by their keypoint predictor, as KeypointDeformer can not be conditioned on arbitrarily positioned control points. Yet, the displacement of each control-point is still annotated by the user.
Also, in the case of ARAP~\cite{10.2312:SGP:SGP07:109-116}, a mesh vertex must be chosen as a control point. 
Therefore, we simply use the nearest vertex to each of the user annotated control points, when comparing to ARAP.

\subsection{Shape Deformation}
\label{ss:shapeDeformation}
We demonstrate the advantages of our method by comparing to both classic and learning-base methods, and evaluate the results by qualitative visualizations and quantitative metrics. We also conducted a user study which suggested that people perceive our deformations as more realistic.

\textbf{Qualitative Evaluation.} Figure~\ref{fig:ShapeDeformationComparison} contains visual comparisons to plain MLS~\cite{Zhu20073DDU}, ARAP~\cite{10.2312:SGP:SGP07:109-116} and KeypointDeformer~\cite{jakab2021keypointdeformer}.
Our method is able to better preserve local features of the source shape, while still adhering the control points guidance across all categories.
The control point displacements are manually selected to reflect reasonable shape deformations,
such as starching the chair seat, changing the wings position of the airplane or bowing the armadillo. Figure~\ref{fig:ShapeDeformationComparison} also demonstrate the limitations of previous approaches.
For the armadillo shape (bottom row), KPD~\cite{jakab2021keypointdeformer} result can not be obtained,
as the armadillo shape does not fall into any of the categories KPD was trained on.
\begin{figure*}
    \centering
    \newcommand{\pl}{77}
    \begin{overpic}[width=\textwidth]{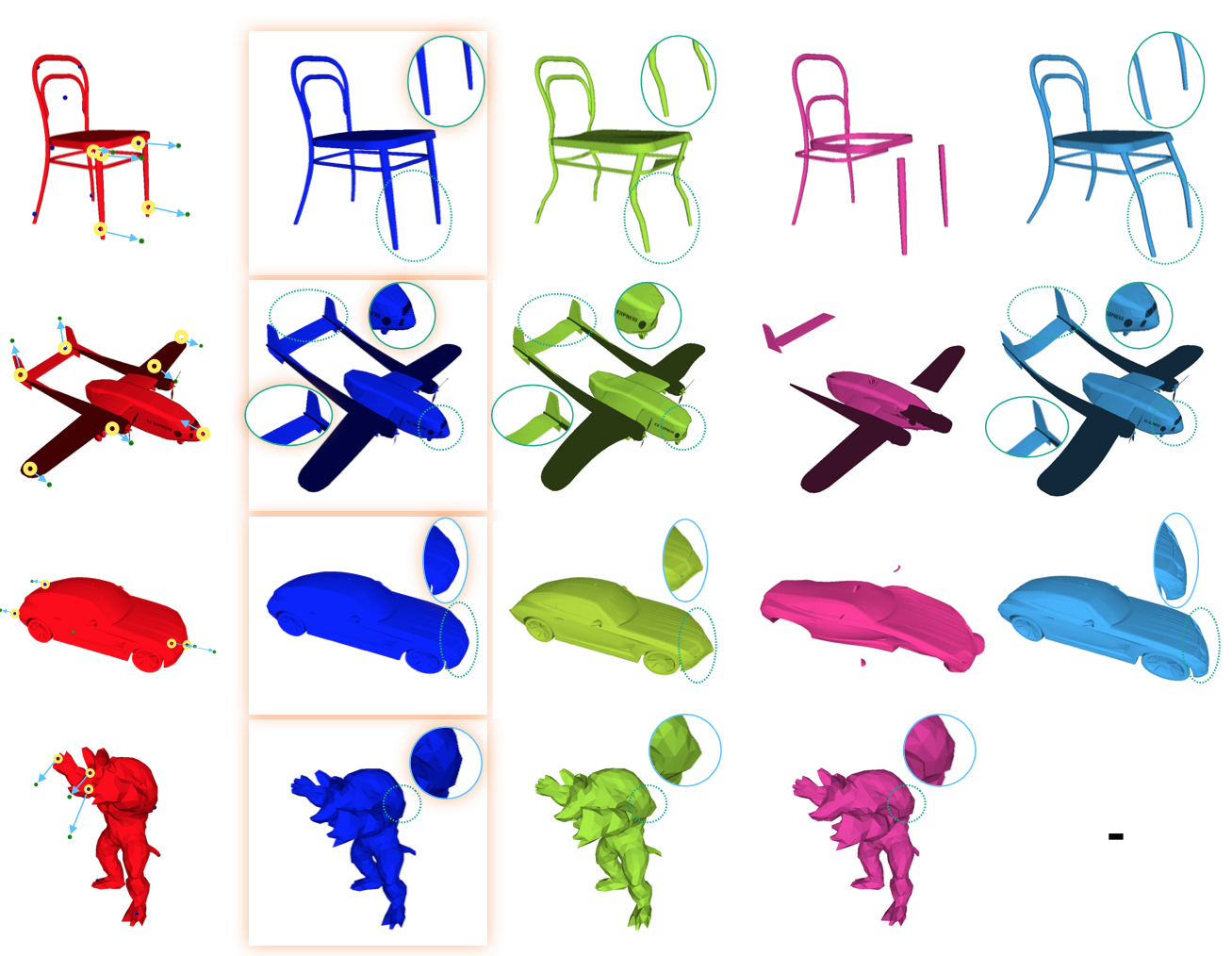}
    \put(6,  \pl){{Source}}
    \put(29,  \pl){{Ours}}
    \put(47,  \pl){{MLS\cite{Zhu20073DDU}}}
    \put(68,  \pl){{ARAP\cite{10.2312:SGP:SGP07:109-116}}}
    \put(87,  \pl){{KPD\cite{jakab2021keypointdeformer}}}
    
    \end{overpic}

    \caption{Comparison of our method to other control point based deformation methods. Our method achieves desirable results compared to traditional and learning-based deformation techniques. 
    Note that the first three rows are non-manifold meshes, which ARAP does not support. Since KPD requires a dataset to be trained on, results for the "armadillo" shape (bottom row) cannot be obtained.}
    \label{fig:ShapeDeformationComparison}
\end{figure*}
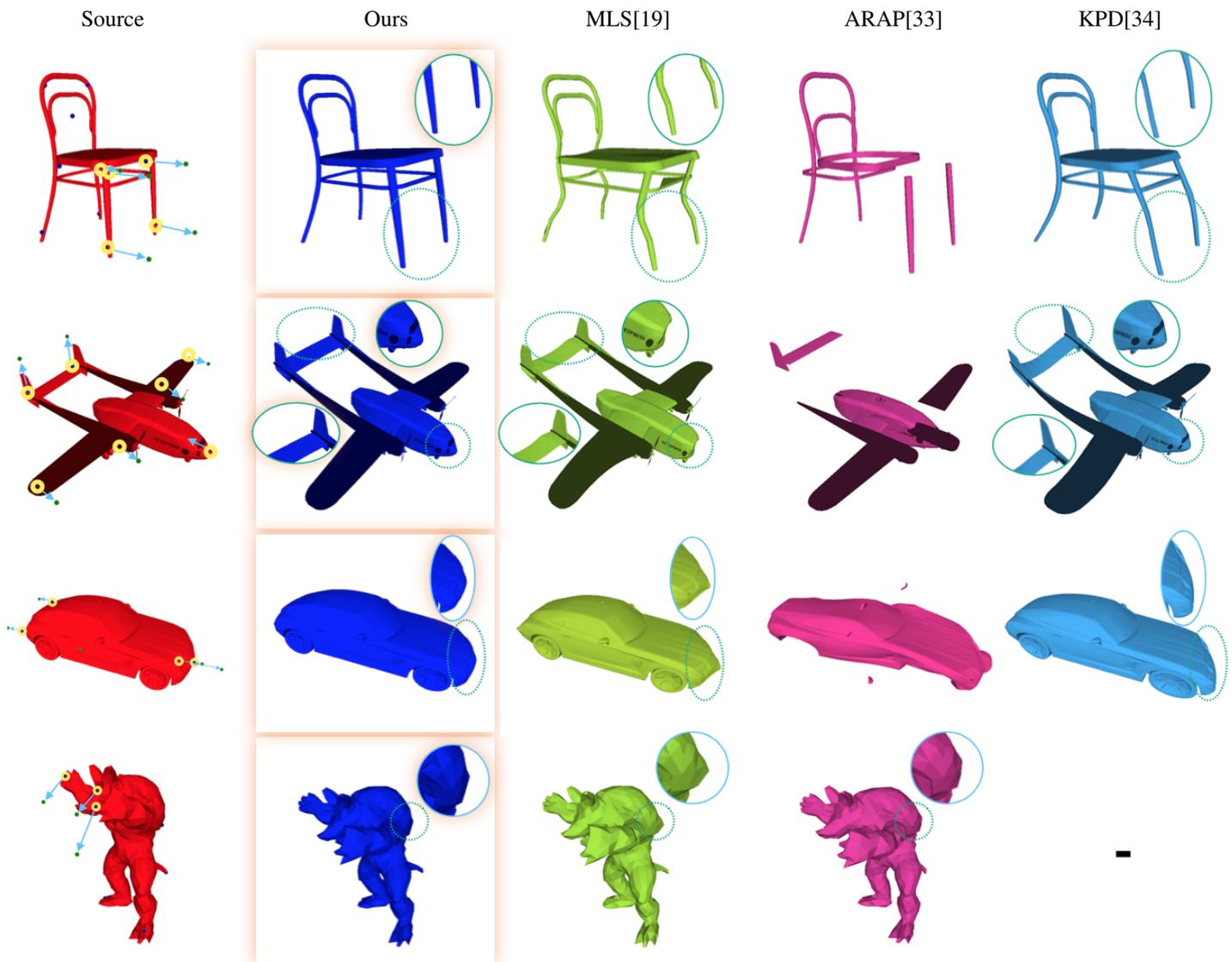
Additional results and comparison to plain MLS can be seen in Figure~\ref{fig:ShapeDeformationResults},
where it shows that MLS suffers from local artifacts.
We claim that the artifacts in Figure~\ref{fig:ShapeDeformationResults} are due to the very high weights, experienced by surface points that are close to the input control points, as further analysis shows in Section~\ref{ss:piecewieseDeformation}.
\begin{figure}
    \centering
    \newcommand{\pl}{101}
    \begin{overpic}[width=\columnwidth]{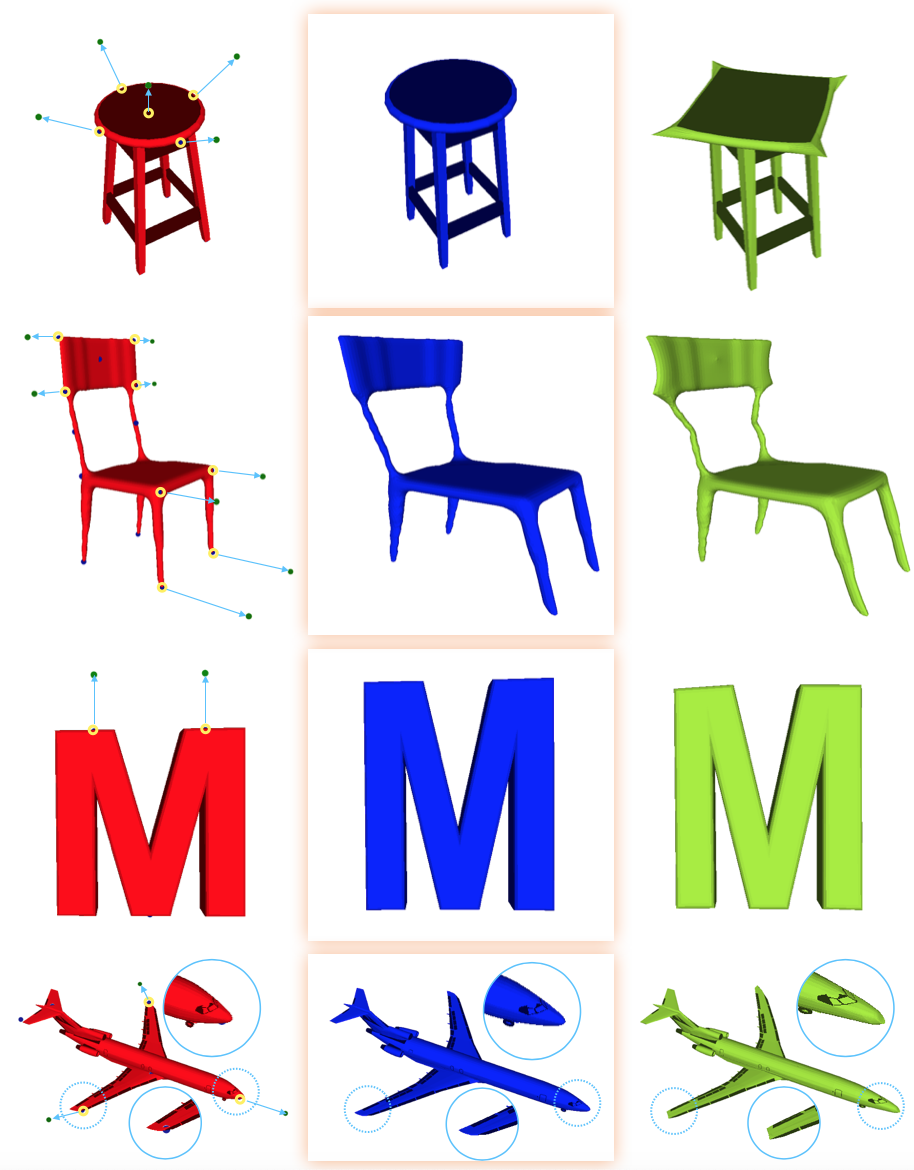}
    \put(9,  \pl){{Source}}
    \put(36,  \pl){{Ours}}
    \put(60,  \pl){{MLS\cite{Zhu20073DDU}}}
    \end{overpic}
    \caption{Deforming source shapes from the left column with user annotated control points, visualized by the blue arrows.
    Observe how MLS suffers from local artifact near the input control points, while our method is able to avoid these issues.
    }
    \label{fig:ShapeDeformationResults}
\end{figure}

\textbf{Quantitative Evaluation.} 
To evaluate the quality of our deformations quantitatively, we measured the distortion in discretized Laplacian magnitude and mean curvature of the deformed shapes, versus the source shapes, in Figure~\ref{fig:ShapeDeformationComparison}.
These measures infer the amount of details preserved on the deformed shape,
and therefor are desirable for a visually appealing deformation.
We used the difference in the Laplacian magnitude rather than the actual Laplacian as we want to be invariant to rotations, that are a desirable property of the deformation, and only want to penalize loss of details that is expressed through the Laplacian magnitude.
As one trivial way to achieve zero distortion in the above measure, is to not deform the shape at all, we also compare the approximation level of the deformed control points to their given displacements.
The results are displayed in Figure~\ref{fig:lap_vs_l2} and Table~\ref{table:quan_res}.

\textbf{User Study.}
We conducted a user study in order to evaluate how realistic our deformations look compared to MLS and KPD, ARAP was excluded on the study, as the deformations failed due to bad triangulation and multiple connected components of the input mesh. 
The user study was done over three different categories (Car, Chair, Airplane), each category contained three different shape,
and three different deformation were applied on each shape (a total of 27 different deformations).
We asked users to choose the most realistic-looking deformation, among the different deformation methods applied on each set of shape and control point configuration.
23 different users replied to our study and the results are displayed in Table~\ref{table:user_study}. The full user study is attached in the supplementary material.

\begin{table}
\centering
\caption{Quantitative evaluation of our method in comparison to other approaches, using the average mean curvature difference between the source shape and the deformed shape. ARAP results for the top 3 shapes are omitted as they are meaningless due to the deformations failing, as can be seen in Figure~\ref{fig:ShapeDeformationComparison}. }
\begin{tabular}{c||c|c|c|c} 
 & \multicolumn{4}{c}{Average Mean Curvature ($\downarrow$)} \\
Shape & Ours & MLS & KPD & ARAP \\
\hline\hline 
Chair     & $\mathbf{16}$ & 19 & 42 & -  \\
\hline
Airplane  & $\mathbf{1764}$ & 1880 & 2142 & - \\
\hline
Car       & $\mathbf{5.2}$ & 8.9 & 17.8 & -  \\
\hline
Armadillo & $\mathbf{0.412}$ & 0.578 & - & 0.574 \\
[1ex] 
\hline 
\end{tabular}
\label{table:quan_res}
\end{table}

\begin{table}
\centering
\caption{Results of our user study are presented in the table below. Users were asked to choose the deformation that looks the most realistic, across results produce by \name, MLS and KPD.
The numbers in the table represent the \% of users that chose the result corresponding to each method.
The study concludes that our method is more likely to produce realistic looking results,
compared to MLS and KPD, as indicated by more users choosing \name for every category in the study.}
\begin{tabular}{c||c|c|c} 
 & \multicolumn{3}{c}{\% of users $(\uparrow)$} \\
Category & Ours & MLS & KPD \\
\hline\hline 
Chair     & $\mathbf{77}$\% & 14\% & 9\% \\ \hline
Airplane  & $\mathbf{56}$\% & 12\% & 32\% \\ \hline
Car       & $\mathbf{71}$\% & 16\% & 13\%  \\
[1ex] 
\hline 
\end{tabular}
\label{table:user_study}
\end{table}

\begin{figure}
    \centering
    \includegraphics[width=\columnwidth]{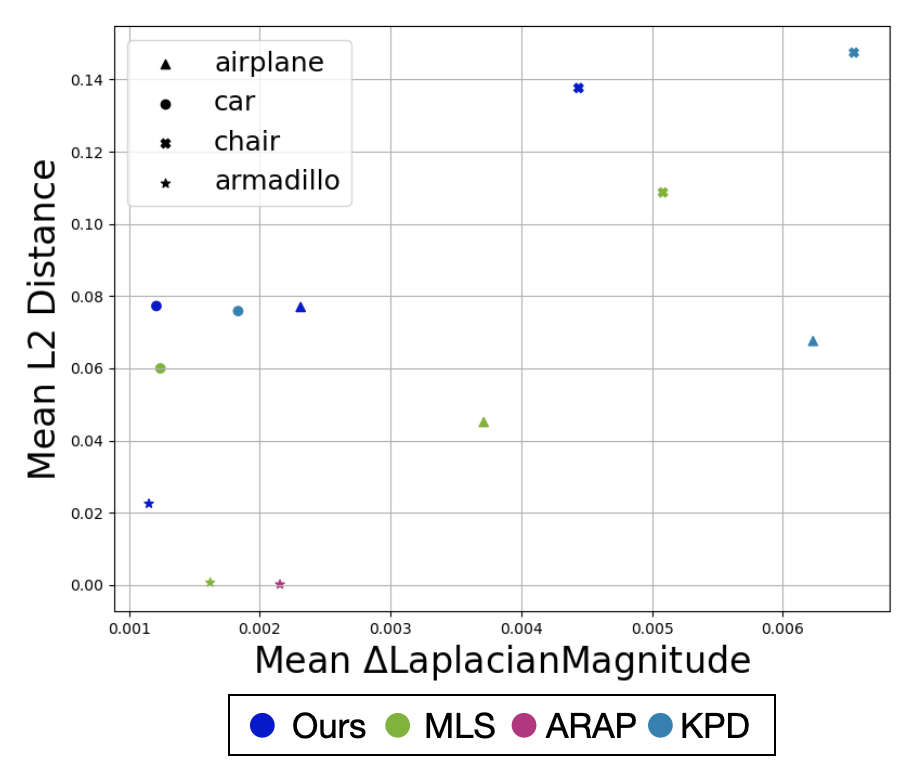}
    \caption{Quantitative comparison of our method against other approaches from Figure~\ref{fig:ShapeDeformationComparison}. Each point is embedded according to its distortion (mean $\Delta$LaplacianMagnitude) and approximation degree of the control points (Mean L2 Distance).
    Point near the bottom left corner are considered as "better".
    ARAP results are omitted for the Chair, Airplane and Car shapes as the deformations failed.
    Notice how our approach is able to achieve the lowest distortion across all categories, while still deforming the shape according to the specified control points.}
    \label{fig:lap_vs_l2}
\end{figure}

\subsection{Piecewise Smooth Deformation}
\label{ss:piecewieseDeformation}
Piecewise smoothness is a desired property for deformations,
as it allows to apply arbitrarily different manipulation to different parts of the shape,
while keeping a naturally and smooth interpolation within and between each part.
The visualizations depicted in Figure~\ref{fig:piecewise_smooth}, show the MLS and \name weights and resulting deformation functions, and demonstrate the struggle of plain MLS in producing piecewise solutions.
The Euclidean-based weighting function starts off at a very high value, close to the control points (bottom-left),
which causes the local artifacts that are observed on the sharp chair edges in Figure~\ref{fig:ShapeDeformationResults} (top row).
The weighting function then decrease rapidly as we move further away from the control point until reaching a non-zero plateau. This result in deformation function (top-left) that is both less smooth (sharp gradient near the control points) and less piecewise (the deformation interpolation between two points is non linear).
Our method, on the other hand, produces weights that are much smaller near control points (bottom-right) ,i.e. achieving an approximation instead of interpolation, then slowly and smoothly decrease until reaching a negligible value at some point, creating a more piecewise solution (top-right).

\begin{figure}
    \centering
    \newcommand{\pl}{96}
    \begin{overpic}[width=\columnwidth]{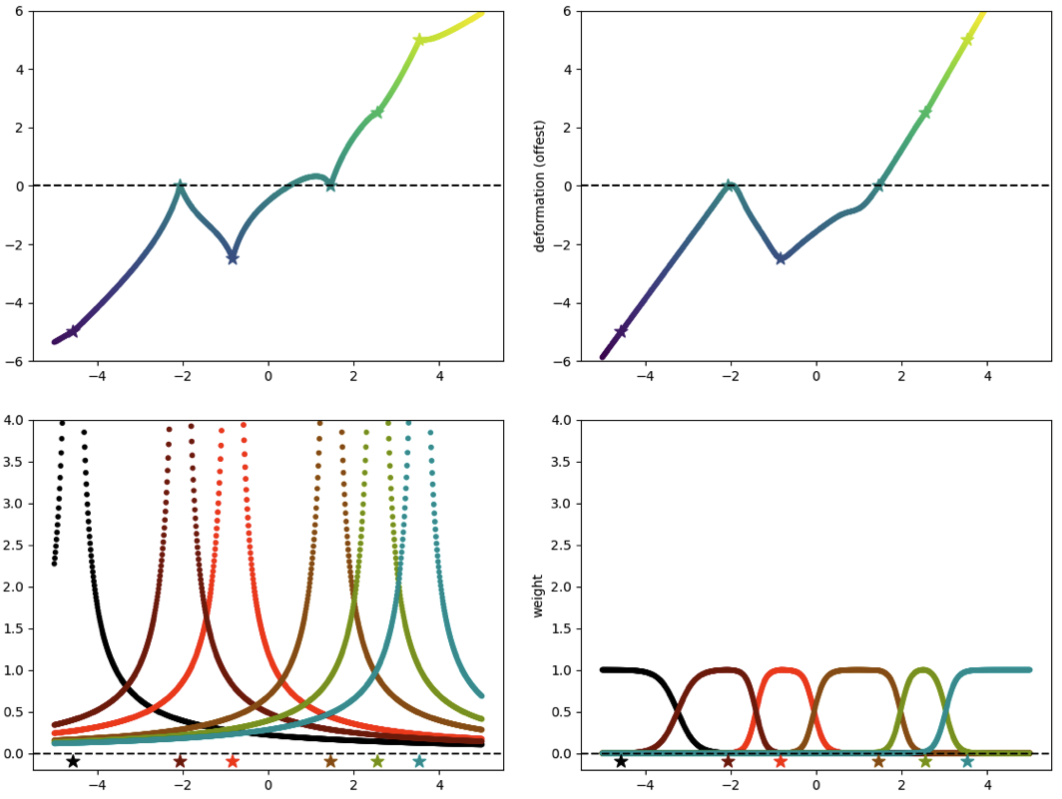}
    \put(18,  76){\textcolor{black}{MLS\cite{Zhu20073DDU}}}
    \put(73,  76){\textcolor{black}{Ours}}
    \put(-4, 14){\rotatebox{90}{\textcolor{gray}{\small weight}}}
    \put(-4, 46){\rotatebox{90}{\textcolor{gray}{\small deformation}}} 
    \put(14,  -3){\textcolor{gray}{\small point location}}
    \put(66,  -3){\textcolor{gray}{\small point location}}
    \end{overpic}
    \caption{The resulting deformation functions (top) and calculated weights (bottom) of our method (right) and MLS (left).
    This illustrates our method's ability to produce more piecewise smooth deformations compared to MLS, as our weight functions are more evenly distributed over each control point adjacent region, where adjacent is with respect to the entire control point configuration, compared to the sharp peaks produced by MLS.}
    \label{fig:piecewise_smooth}
\end{figure}
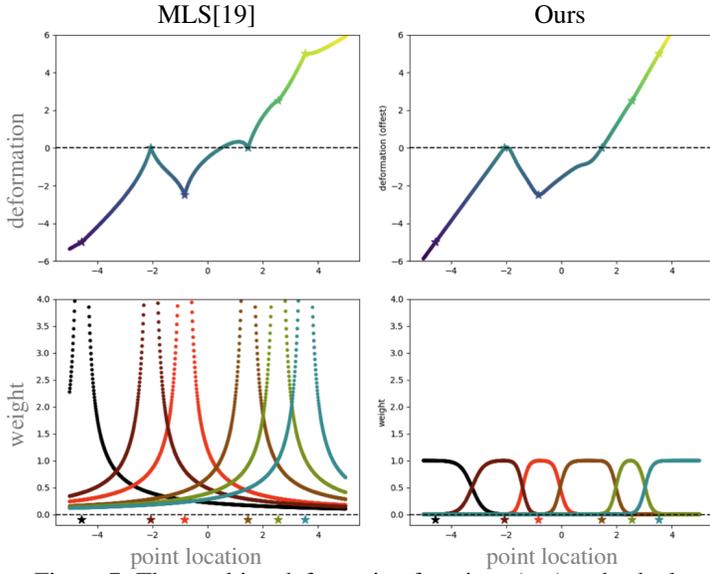

\subsection{Approximating vs. Interpolating}
\label{ss:ApproVsInter}
We showcase our technique's ability to enable users to control the degree of approximation v.s. interpolation, using a single temperature parameter.
A low temperature value results in sharper weight assignments, mainly considering the closest control point, similar to nearest neighbour interpolation of the input displacements.
On the other hand, a high temperature value results in a smoother combination of the control point specified displacements, i.e. approximation, as the resulting displacement at each control point does not necessarily equal the specified user input exactly.
Changing the temperature value does not require additional training, as it only amounts to scaling the trained network outputs. Figure~\ref{fig:temperature_effect} presents deformation results under different temperatures.
The control points and their displacements are displayed to facilitate the visualization of the approximation quality.
\begin{figure*}
    \centering
    \includegraphics[width=\textwidth]{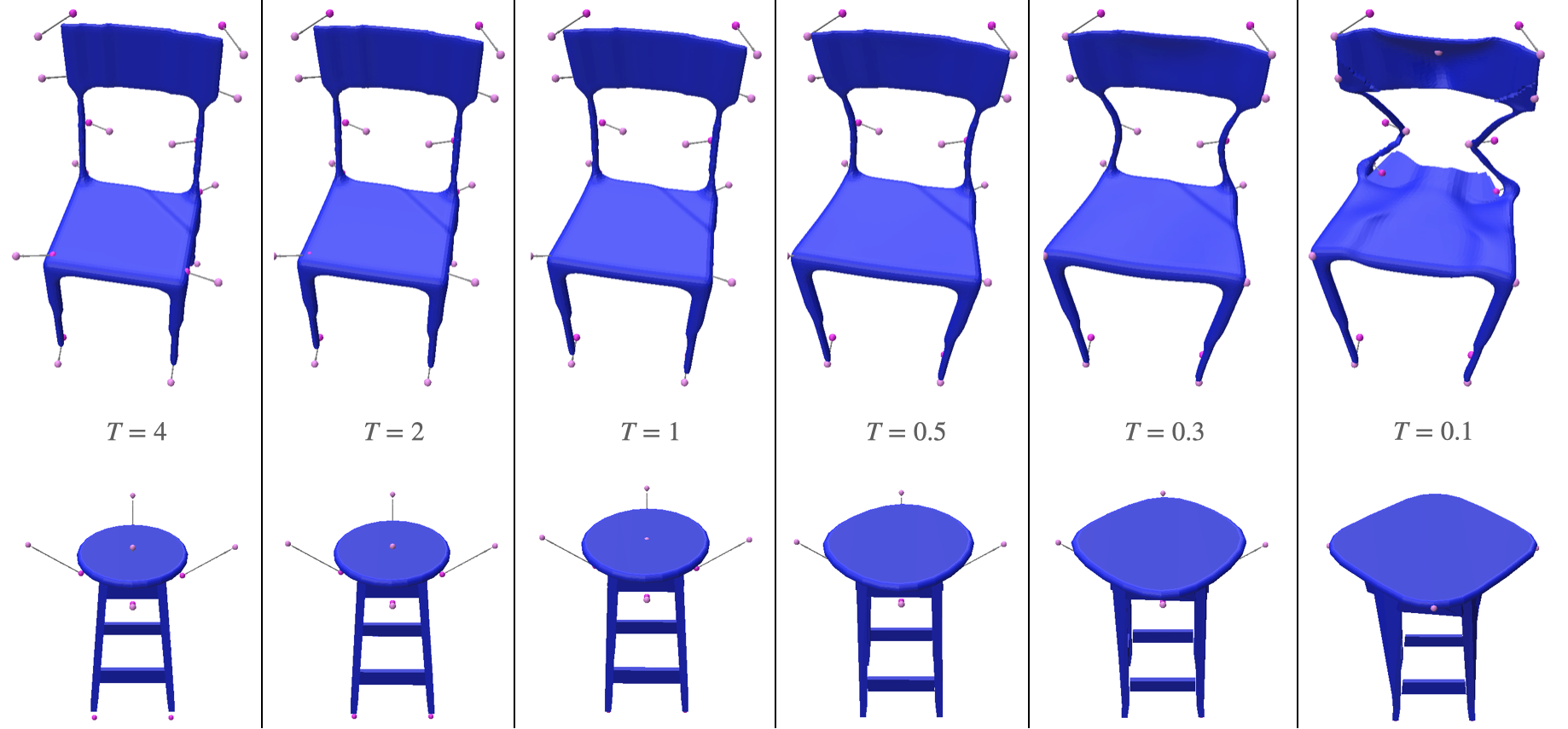}
    \caption{The temperature hyper-parameter effect demonstrated on various chairs.
    The temperature value gives the user control over the approximation level of the deformation, as a trade off between approximation and interpolation of the input control point guidance.}
    \label{fig:temperature_effect}
\end{figure*}



\subsection{MLS Ablation}
\label{ss:ablation}
\textbf{MLS hyper-parameters.} MLS \cite{Zhu20073DDU} allows the user a "fall-out" hyperparameter, i.e., $\alpha$, to control the \textit{typical affecting distance} of each control point.
Figure~\ref{fig:MLS_ablation_alpha} shows that even though $\alpha$ has a smoothing effect on the deformation, the results still suffers from significant artifacts.
This is evident in the collateral damage caused by high influence of irrelevant control points, seen on the legs of the bottom row left most chair, and on the back rest of the top row chair.
This phenomenon can also be foreseen mathematically from Equation~\ref{eq:2},
as changing the $\alpha$ parameter does not resolve the large weights and gradient norms near control points, also discussed in Section~\ref{ss:piecewieseDeformation}.

\begin{figure*}
    \centering
    \includegraphics[width=\textwidth]{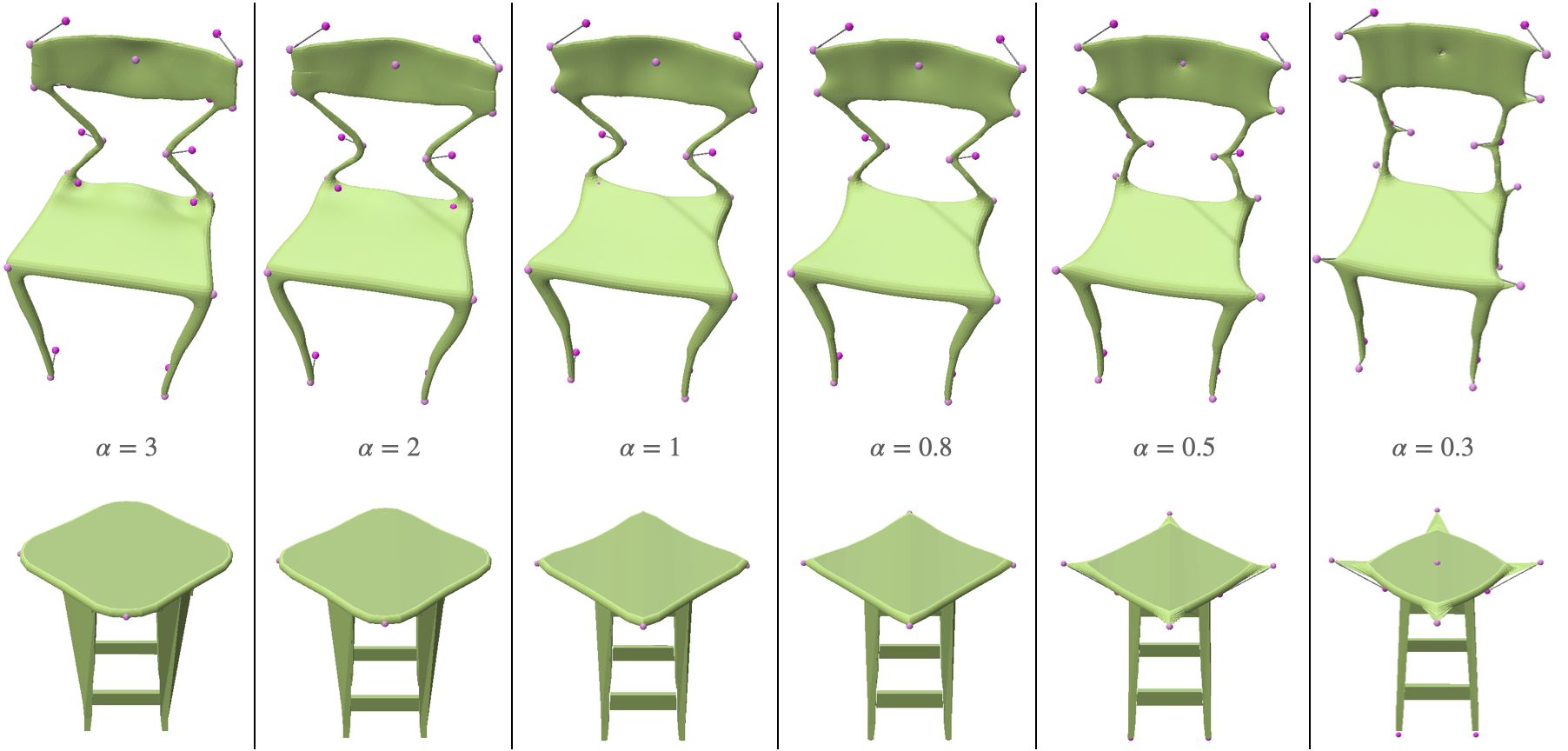}
    \caption{The $\alpha$ hyper-parameter effect in the plain MLS method.
    This parameter provides the user control over the region of influence of each control point. Observe how altering this parameter, even though mitigating certain artifacts, such as sharp edges, causing new ones, e.g. bending rigid legs in the top chair or incompatible control point influence, manifested in the legs of the bottom leftmost chair.}
    \label{fig:MLS_ablation_alpha}
\end{figure*}

Another approach for trying to achieve a better weighting functions based on an Euclidean distance, is to add a small value $\epsilon$, to the denominator of the weighting function, i.e., $w_i(\mathbf{x})=\frac{1}{d(\mathbf{p_i},\mathbf{x})^{2\alpha}+\epsilon}$.
Figure~\ref{fig:MLS_ablation_epsilon} reveals that adding $\epsilon$ does not sufficiently resolve local artifacts and undesirable deformation properties, such as the bending of the legs in the top row chair.

\begin{figure*}
    \centering
    \includegraphics[width=\textwidth]{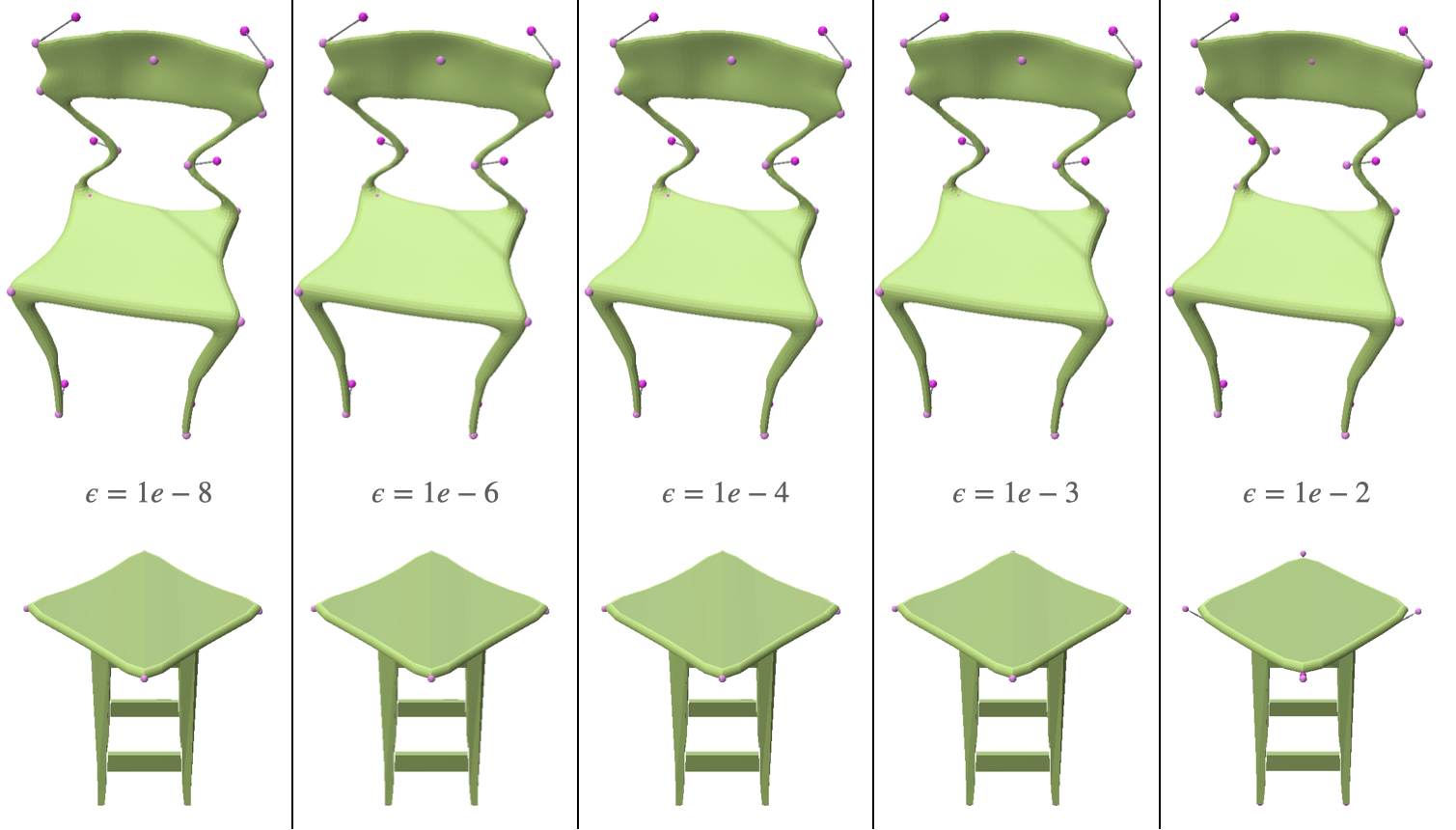}
    \caption{Demonstrating the effect of adding an $\epsilon$ parameter to the denominator of the weighting function of plain MLS on the chair shape.
    This value is used for numerical stability but can also, as a by-product, enable the user to smooth out the deformation of the shape close to control points. Observe that although indeed mitigating the sharp artifacts the deformation still suffers from them, as well as from other problems, such as bending of rigid parts (top chair legs).}
    \label{fig:MLS_ablation_epsilon}
\end{figure*}



\section{Conclusions}

We have presented a geometry-aware space-deformation technique that is based on the MLS framework. The key idea is to leverage the power of neural networks to learn a weighting function associated with the given control points. The learned weights adhere to the geometric configuration of the control points which implicitly respects the underlying geometry, and results in piecewise smoothness.

Traditional MLS treats the control point displacements as \textit{hard-constraints}. On the other hand, our \name treats the control point displacements as \textit{soft-constraints} in order to obtain piecewise-smooth deformations.
Indeed, there is an inherent trade-off between adherence to the constraints and a piecewise-smooth weighting function. Our framework provides a relaxation parameter which can trade-off smoothness for constraint-adherence, which enables intuitive interactive manipulation of the shape.  

Our framework is built for an interactive scenario with user-specified control points and constraints.
In the future, we would like to train a network to learn the location and number of the control points based on an analysis of the input shape. We would like also to consider grouping and structuring the control points to reflect the symmetries and relationships of the input shape, such that the editing will be faithful to the input shape semantics.



\bibliographystyle{unsrt}  
\bibliography{references}

\end{document}